\documentclass{article}
\usepackage[accepted]{icml2020}

\usepackage[utf8]{inputenc}
\usepackage{amssymb,amsfonts,amsmath}
\usepackage{amsthm}
\usepackage{comment}
\usepackage{bm}
\usepackage{bbm}
\usepackage{mathtools}
\usepackage{breqn}
\usepackage{svg}
\usepackage{makecell}
\usepackage{booktabs}

\usepackage{scalerel,stackengine}
\stackMath
\newcommand\reallywidehat[1]{%
\savestack{\tmpbox}{\stretchto{%
  \scaleto{%
    \scalerel*[\widthof{\ensuremath{#1}}]{\kern-.6pt\bigwedge\kern-.6pt}%
    {\rule[-\textheight/2]{1ex}{\textheight}}
  }{\textheight}%
}{0.5ex}}%
\stackon[1pt]{#1}{\tmpbox}%
}

\newtheorem{theorem}{Theorem}
\newtheorem{corollary}{Corollary}[theorem]
\newtheorem{lemma}{Lemma}

\newtheorem{definition}{Definition}

\DeclarePairedDelimiter{\ceil}{\lceil}{\rceil}
\DeclareMathOperator{\var}{var}
\DeclareMathOperator{\E}{\mathbb{E}}
\DeclareMathOperator{\I}{\mathbbm{1}}
\DeclareMathOperator{\similarity}{sim}

\newcommand\CMS{\mathrm{CMS}}

\begin{document}

\twocolumn[
\icmltitle{Sub-linear Memory Sketches for Near Neighbor Search on Streaming Data}

\begin{icmlauthorlist}
\icmlauthor{Benjamin Coleman}{riceECE}
\icmlauthor{Richard G Baraniuk}{riceECE,riceCS}
\icmlauthor{Anshumali Shrivastava}{riceECE,riceCS}
\end{icmlauthorlist}

\icmlaffiliation{riceCS}{Department of Computer Science, Rice University, Houston, Texas, USA}
\icmlaffiliation{riceECE}{Department of Electrical and Computer Engineering, Rice University, Houston, Texas, USA}

\vskip 0.3in
]

\icmlcorrespondingauthor{Benjamin Coleman}{ben.coleman@rice.edu}
\icmlcorrespondingauthor{Anshumali Shrivastava}{anshumali@rice.edu}
\printAffiliationsAndNotice{}

\begin{abstract}

We present the first sublinear memory sketch that can be queried to find the nearest neighbors in a dataset. Our online sketching algorithm compresses an N element dataset to a sketch of size $O(N^b \log^3 N)$ in $O(N^{(b+1)} \log^3 N)$ time, where $b < 1$. This sketch can correctly report the nearest neighbors of any query that satisfies a stability condition parameterized by $b$. We achieve sublinear memory performance on stable queries by combining recent advances in locality sensitive hash (LSH)-based estimators, online kernel density estimation, and compressed sensing. Our theoretical results shed new light on the memory-accuracy tradeoff for nearest neighbor search, and our sketch, which consists entirely of short integer arrays, has a variety of attractive features in practice. We evaluate the memory-recall tradeoff of our method on a friend recommendation task in the Google Plus social media network. We obtain orders of magnitude better compression than the random projection based alternative while retaining the ability to report the nearest neighbors of practical queries. 

\end{abstract}





\section{Introduction}
Approximate near-neighbor search (ANNS) is a fundamental problem with widespread applications in databases, learning, computer vision, and much more~\cite{gionis1999similarity}. Furthermore, ANNS is the first stage of several data processing and machine learning pipelines and is a popular baseline data analysis method. Informally, the problem is as follows. Given a dataset $\mathcal{D} = {x_1, \ x_2, ...,\  x_N}$, observed in a one pass sequence, build a data structure $\mathcal{S}$ that can efficiently identify a small number of data points $x_i \in \mathcal{D}$ that have high similarity to any dynamically generated query $q$. 

In this paper, we focus on low-memory ANNS in settings where it is prohibitive to store the complete data in any form. Such restrictions naturally arise in extremely large databases, computer networks, and internet-of-things systems~\cite{johnson2017billion}. We want to compress the dataset $\mathcal{D}$ into a sketch $\mathcal{S}$ that is as small as possible while still retaining the ability to find near-neighbors for any query. Furthermore, the algorithm should be one pass as the second pass is prohibitive when we cannot store the full data in any form. It is common wisdom that the size of $\mathcal{S}$ must scale linearly ($\geq \Omega(N)$), even if we allow algorithms that only identify the locations of the nearest neighbors. In this work, we challenge that wisdom by constructing a sketch of size $O(N^b \log^3{N})$ bits in $O(N^{b+1} \log^3{N})$ time. Our sketch can identify near-neighbors for {\em stable} queries with high probability in $O(N^{b+1} \log^3{N})$ time. The value of $b$ depends on the dataset, but $b$ can be significantly less than 1 for many applications of practical importance. It should be noted that our sketch does not return the near neighbors themselves, since we do not store the vectors in any form. Instead, we output the identity or the index of the nearest neighbor, which is sufficient for most applications and does not fundamentally change the problem. Our sketch also does not attempt to correctly answer every possible near-neighbor query in sublinear memory, as this would violate information theoretic lower bounds. Instead, we provide a graceful tradeoff between the stability of a near neighbor search query and the memory required to obtain a correct answer. 

\subsection{Our Contribution}

Our main contribution is a one-pass algorithm that produces a sketch $\mathcal{S}$ that solves the exact $v$-nearest neighbor problem in sub-linear memory with high probability. A formal problem statement is available in Section~\ref{sec:problemstatements} and our theoretical results are formally stated in Section~\ref{sec:theory}. Our algorithm requires $O(N^{b+1} \log^3{N})$ time to construct $\mathcal{S}$ and the same time to return the $v$ nearest neighbors for a query. Here, $b$ is a query-dependent value that describes the stability or difficulty of the query. Our guarantees are general and work for any query, but the sketch is only sub-linear when $b < 1$. In practice, one commits to a given $b$ value and obtains the guarantees for all queries satisfying our conditions. 

We obtain our sketch by merging compressed sensing techniques with recently-developed sketching algorithms. Surprisingly, we find that the hardness of a near-neighbor query is directly related to the notion of sparsity, or signal-to-noise ratio (SNR), in compressed sensing ~\cite{donoho2006compressed,tropp2007signal}. This connection allows us to analyze geometric structure in the dataset using the very well-studied compressed sensing framework. The idea of exploiting structure to improve theoretical guarantees has recently gained traction because it can lead to stronger guarantees. For instance, the first improvements over the seminal near-neighbor search results of \cite{indyk1998approximate} were obtained using \textit{data-dependent} hashing~\cite{andoni2014beyond}. These methods use information about the data distribution to generate an optimal hash for a given dataset. In this work, we assume that the dataset has a set of general properties that are common in practice and we construct a data structure that exploits these properties. In general, the communication complexity of the near neighbor problem is $O(N)$. Our method requires sub-linear memory because our data assumptions limit the set of valid queries. 

We support our theoretical findings with real experiments on large social-network datasets. Our theoretical techniques are sufficiently general to accommodate a variety of compressed sensing methods and KDE approximation algorithms. However, in practice we apply our theory using the Count-Min Sketch (CMS) as the compressed sensing method and the recently-proposed RACE sketch for KDE~\cite{coleman2020race}. Our RACE-CMS sketch inherits a variety of desirable practical properties from the RACE and CMS sketches that are used in its construction. When implemented this way, our near neighbor sketch consists entirely of a set of integer arrays. Furthermore, RACE sketches are linear, parallel and mergeable, allowing us to realize many practical gains using RACE-CMS. For instance, despite a query time complexity that is theoretically worse than linear search, RACE-CMS can be implemented in such a way that it is fast and practical to construct and query, processing thousands of vectors each second. As a result, we believe that our method will enable a variety of practical applications that need to perform near neighbor search in the distributed streaming setting with limited memory.

\section{Applications}
Here, we describe several applications for low-memory near neighbor sketches. 

\textbf{Graph Compression for Recommendation: } 
In recommendation systems, we represent relationships, such as friendship or co-purchases, as graphs. Given $N$ users, we represent each user as an $N$ dimensional sparse vector, where non-zero entries correspond to edges or connections. To perform recommendations, we often wish to find pairs of users that are mutually connected to a similar set of other users. The process of identifying these users is a similarity search problem over the $N$ dimensional sparse vector representation of the graph~\cite{hsu2006collaborative}. Online graphs can be very large, with billions of nodes and trillions of edges~\cite{ching2015one}. Since graphs at this scale are prohibitively expensive to store and transmit, methods capable of compressing the network into a small and informative sketch could be invaluable for large-scale recommendations. 

\textbf{Robust Caching: } The process of caching previously-seen data is a central component of many latency-critical applications including search engines, computer networks, web browsers and databases. While there are many well-established methods, such as Bloom filters, to detect exact matches, caching systems cannot currently report the distance between a query element and the contents of the cache. Our sketches can be used to implement caching mechanisms that are robust to minor perturbations in the query. Such a capability naturally provides better anomaly detection, robust estimation and retrieval. Since similar data structures can fit into the cache of modern processors~\cite{luo2018arrays}, our sketches could be an effective practical tool for online caching algorithms. 

\textbf{Distributed Data Streaming: } In application domains such as the internet-of-things (IoT) and computer networks, we often with to build classifiers and other machine learning systems in the streaming setting~\cite{ma2009identifying}. In practice, sketching is a critical component of distributed data collection pipelines. For instance, Apple uses a wide variety of sketches to enable mobile users to transmit valuable information that can be used to train machine learning models while minimizing the data transmission cost~\cite{apple2017}. Similar challenges occur with distributed databases and IoT settings, where data generators can be scattered across a network of connected devices. Such applications require sketching methods to minimize the data communication cost while preserving utility for downstream learning applications. Since our sketches consist of integer arrays, they can easily be serialized and sent over a network.

\subsection{Related Work}
\begin{table*}[t]
\caption{Summary of related work. Results are shown for a $d$-dimensional dataset of $N$ points. Existing methods \cite{johnson1984extensions,indyk2018approximate, agarwal2005geometric} can estimate distances to all points in the dataset with a $1 \pm \epsilon$ multiplicative error (full $\epsilon$ dependence not shown). Our method estimates the similarity with all points having a $\pm \epsilon$ additive error, where $b$ depends on the properties of the dataset.}
\begin{center}
\begin{tabular}{ l | l | l | l }
\hline
Method & Sketch Size (bits) & Sketch Time & Comments\\
\hline
\makecell[l]{No compression} & $d N \log N$ & N/A & - \\
\hline
\makecell[l]{Random projections~\cite{johnson1984extensions}} & $N \log^2{N}$ & $N \log{N}$ & \makecell[l]{Widely used in practice} \\
\hline
\makecell[l]{Compressed clustering tree~\cite{indyk2018approximate}} & $N \log{N}$ & $ N d \log^{O(1)}{N}$
& \makecell[l]{Multiple passes}\\
\hline
\makecell[l]{Coresets~\cite{agarwal2005geometric}} & $d \epsilon^{-(d-1)}$ & $N + \epsilon^{-(d-1)}$ & \makecell[l]{Multiple Passes} \\
\hline
\makecell[l]{This work} & $N^b \log^3{N}$ &$N^{b+1} \log^3{N}$ & \makecell[l]{$b < 1$ for stable queries} \\
\hline
\end{tabular}
\label{table:relatedwork}
\end{center}
\end{table*}

The problem of finding near-neighbors in {\em sub-linear time} is a very well-studied problem with several solutions~\cite{indyk1998approximate}. However, the {\em memory requirement} for near-neighbor search has only recently started receiving attention~\cite{indyk2018approximate,indyk2017near}. Although hueristic methods for sample compression are employed in practice, the best theoretical result in this direction requires $O(N \log{N})$ memory and therefore does not break the linear memory bound~\cite{indyk2018approximate}. Table~\ref{table:relatedwork} contains a summary of existing work in the area. To the best of our knowledge, the algorithm described in this paper is the first to perform near-neighbor search using asymptotically sub-linear memory. 

\paragraph{Coresets or Clustering Based Approaches:} A reasonable compression approach is to construct a coreset or represent the dataset as a set of clusters. For instance, the widely-used FAISS system compresses vectors using product quantization~\cite{jegou2010product}. There are also sampling procedures to construct a subset $P$ of $\mathcal{D}$ and guarantee the existence of a point $p \in P$ such that $d(p,x) < \epsilon$ for $\epsilon > 0$. The cluster-based approach from ~\cite{har2014down} uses similar ideas to reduce the space for $v$-nearest-neighbor by a constant factor of $\frac{1}{v}$. However, our procedure is superior in the following two regards. First, coresets and sample-based compression methods require parallel access to the entire dataset at once to determine which points to retain in the sketch. As an example, the sketch in~\cite{har2014down} requires an offline clustering step. Therefore, it is impossible to stream queries to the sketch efficiently using existing methods. Second, cluster approximations of the data cannot solve the exact $v$-nearest neighbor problem because the sketching process removes points from the dataset. Despite the guarantees that can be obtained using $\epsilon$ coverings of the dataset, there may be any number of near-neighbors within $\epsilon$ of the query that have been discarded during sketching. 


Perhaps most importantly, our method requires weaker assumptions about the dataset. Cluster-based methods assume that the dataset has a clustered structure that can be approximated by a small collection of centroids. To achieve high compression ratios, coreset methods require similar assumptions. However, our method is valid even when there is no efficient cluster representation. Our weak assumptions are particularly applicable to recent problems in recommendation systems, graph compression and neural embedding models. In this context, we are given a dataset where each embedding or object representation is close to a relatively small number of other elements. Furthermore, we expect most of our queries to be issued in regions that contain only a few elements from the dataset. Although there may be no large-scale hierchical clustering structure, our method can exploit the weaker structure in the dataset to provide good compression without the need for complex clustering and sample compression algorithms. 

Finally, we note that our approach is much simpler to understand and analyze than existing methods. While clustering methods can achieve good performance, they usually require complex distance-approximation methods at query time. Sketch construction consists of computationally-intensive clustering steps or coreset sampling routines that have many moving parts. In contrast, our data structure is a simple array of integer counters with a fixed size. Therefore, we expect that our method will be attractive to practitioners and system designers.





\subsection{Background}

Our algorithm uses recent advances in locality-sensitive hashing (LSH)-based sketching with standard compressed sensing techniques. Before covering our method in detail and presenting theoretical results, we briefly review some useful results in sketching and compressed sensing. 

\subsection{Problem Statement}
\label{sec:problemstatements}
In this paper, we solve the exact $v$-nearest neighbor problem. The $v$-nearest neighbor problem is to identify all of the $v$ closest points to a query with high probability. The difficulty of the $v$-nearest neighbor problem is data-dependent. To capture the difficulty of a query, we use the notion of near-neighbor stability from the seminal paper~\cite{beyer1999nearest}. 

\begin{definition} Exact $v$-nearest neighbor \\
\label{def:vNN}
Given a set $\mathcal{D}$ of points in a $d$-dimensional space and a parameter $v$, construct a data structure which, given any query point $q$, reports a set of $v$ points in $\mathcal{D}$ with the following property: Each of the $v$ nearest neighbors to $q$ is in the set with probability $1-\delta$.
\end{definition}

\begin{definition} Unstable near-neighbor search\\
\label{def:unstableNN}
A nearest neighbor query is unstable for a given $\epsilon$ if the distance from the query point to most data points is $\le (1 + \epsilon)$ times the distance from the query point to its nearest neighbor.
\label{def:unstableneighbor}
\end{definition}

\subsection{Compressed Sensing and the Count Min Sketch}
\label{sec:compressedsensing}

Compressed sensing is the area in signal processing that deals with the recovery of compressible signals from a sublinear number of measurements. The task is to recover an $N$-length vector $\mathbf{x}$ from a vector $\mathbf{y}$ of $M$ linear combinations, or measurements, of the $N$ components of $\mathbf{x}$. The problem is tractable when $\mathbf{x}$ is $v$-sparse and has only $v$ nonzero elements. For a more detailed description of the compressed sensing problem, see~\cite{baraniuk2007compressive}. The fundamental result in compressed sensing is that we can exactly recover $\mathbf{x}$ from $\mathbf{y}$ using only $M = O(v \log{N/v})$ measurements. 

In the streaming literature, the $v$ nonzero elements are often called \textit{heavy hitters}. The Count-Min Sketch (CMS) is a classical data summary to identify heavy hitters in a data stream. The CMS is a $d\times w$ array of counts that are indexed and incremented in a randomized fashion. Given a vector $\mathbf{s}$, for every element $s_i$ in $\mathbf{s}$, we apply $d$ universal hash functions $h_1(\cdot), ... h_d(\cdot)$ to $i$ to obtain a set of $d$ indices. Then, we increment the CMS cells at these indices. When all elements of $\mathbf{s}$ are non-negative, we have a point-wise bound on the estimated $s_i$ values returned by the CMS~\cite{cormode2005improved}. For the sake of simplicity, we only consider the CMS when presenting our results. Finding heavy hitters is equivalent to compressed sensing~\cite{indyk2013sketching}, and there are an enormous number of valid measurement matrices in the literature~\cite{candes2011probabilistic}. Other compressed sensing methods can improve our bounds, but we defer this discussion the supplementary materials.



\begin{theorem}
\label{thm:CMSPointwiseGuarantee}
Given a CMS sketch of the non-negative vector $\mathbf{s} \in \mathbb{R}^{N}_+$ with $d = O\left(\log\left(\frac{N}{\delta}\right)\right)$ rows and $w = O\left(\frac{1}{\epsilon}\right)$ columns, we can recover a vector $\mathbf{s}^{\CMS}$ such that we have the following point-wise recovery guarantee with probability $1-\delta$ for each recovered element $s^{\CMS}_{i}$: 
\begin{equation}
s_i \leq s^{\text{CMS}}_{i}\leq s_i + \epsilon |\mathbf{s}|_{1}
\end{equation}
\end{theorem}


\subsection{Locality-Sensitive Hashing}
\label{sec:LSH}
LSH~\cite{indyk1998approximate} is a popular technique for efficient approximate nearest-neighbor search. An LSH family is a family of functions with the following property: Under the hash mapping, similar points have a high probability of having the same hash value. We say that a collision occurs whenever the hash values for two points are equal, i.e. $h(p) = h(q)$. The probability $\text{Pr}_{\mathcal{H}}[h(p) = h(q)]$ is known as the collision probability of $p$ and $q$. In this paper we will use the notation $p(p,q)$ to denote the collision probability of $p$ and $q$. For our arguments, we will assume a slightly stronger notion of LSH than the one given by~\cite{indyk1998approximate}. We will suppose that the collision probability is a monotonic function of the similarity between $p$ and $q$. That is
\begin{equation}
    p(p,q) \propto f(\similarity(p,q))
\end{equation}
where $\similarity(p,q)$ is a similarity function and $f(\cdot)$ is monotone increasing. LSH is a very well-studied topic with a number of well-known LSH families in the literature~\cite{gionis1999similarity}. Most LSH families satisfy this assumption. 



\subsection{Repeated Array-of-Counts Estimator (RACE)}
\label{sec:ACE}
Recent work has shown that LSH can be used for efficient unbiased statistical estimation~\cite{spring2017new,charikar2017hashing,luo2018arrays}. The RACE algorithm~\cite{coleman2020race} replaces the universal hash function in the CMS with an LSH function. The result is a sketch that approximates the kernel density estimate (KDE) of a query. Here, we re-state the main theorem from~\cite{luo2018arrays} using simpler notation. 

\begin{theorem} ACE Estimator~\cite{luo2018arrays}\\
\label{thm:ace}
Given a dataset $\mathcal{D}$, an LSH function $l(\cdot) \mapsto [1,R]$ and a parameter $K$, construct an LSH function $h(\cdot) \mapsto [1,R^K]$ by concatenating $K$ independent $l(\cdot)$ hashes. Let $A \in \mathbb{R}^{R^K}$ be an array of $O(R^K \log N)$ bits where the $i^{\text{th}}$ component is
$$ A[i] = \sum_{x \in \mathcal{D}} \I_{\{h(x) = i\}}$$
Then for any query $q$, 
$$ \E[A[h(q)]] = \sum_{x \in \mathcal{D}} p(x,q)^K$$
\end{theorem}
We will heavily leverage the observation that $A[h(q)]$ is an unbiased estimator of the summation of collision probabilities. This sum is a kernel density estimate over the dataset~\cite{coleman2020race}, where the kernel is defined by the LSH function.

\begin{algorithm}[H] 
\caption{One-Pass Online Sketching Algorithm}
\label{alg:sketch}
\begin{algorithmic}
\REQUIRE $\mathcal{D}$
\ENSURE $d \times w$ RACE arrays indexed as $A_{i,j,o}$
 \STATE \textbf{Initialize:} $k \times d \times w \times R$ independent LSH family (denoted by $L(\cdot)$) and $d$ independent 2-universal hash functions $h_i(\cdot)$, $i \in [1-d]$, each taking values in range $[1-w]$.
 \WHILE{not at end of data $\mathcal{D}$}
 \STATE read current $x_j$;
  \FOR{$o$ in 1 to $r$} 
  \FOR{$i$ in 1 to $d$}
   \STATE $A_{i,h_i(j),o}[L[x_j]]$++;
  \ENDFOR
  \ENDFOR
 \ENDWHILE
\end{algorithmic}
\end{algorithm}

\begin{algorithm}[H] 
\caption{Querying Algorithm}
\label{alg:query}
\begin{algorithmic}
\REQUIRE Sketch from Algorithm~\ref{alg:sketch}, query $q$
\ENSURE Identities of Top-$v$ neighbors of $q$
 \STATE \textbf{We already have:} $k \times d \times w \times R$ independent LSH family (denoted by $L(\cdot)$) and $d$ independent 2-universal hash functions $h_i(\cdot)$, $i \in [1-d]$, each taking values in range $[1-w]$ from Algorithm~\ref{alg:sketch}.
 \FOR{$i$ in 1 to $d$}
 \FOR{$j$ in 1 to $w$}
   \STATE $\reallywidehat{\CMS}_{(i,j)} = \text{MoM}(A_{i,j,o}[L(q)])$
 \ENDFOR
 \ENDFOR
 
 \FOR{$j$ in 1 to $n$}
   \STATE $\reallywidehat{\mathbf{s}}_j = \reallywidehat{p}(q,x_j)^K = \min_{i}\  \reallywidehat{\CMS}_{(i,h_i(j))}$
 \ENDFOR
 
 \STATE Report top-$v$ indices of $\reallywidehat{\mathbf{s}}$ as the neighbors. 
 
\end{algorithmic}
\end{algorithm}

\section{Intuition}
We propose Algorithm~\ref{alg:sketch} as an online near-neighbor sketching method and Algorithm~\ref{alg:query} to query the sketch. The intuition behind our algorithm is as follows. Consider the naive method to perform near-neighbor search. We begin by finding the pairwise distances between the query and each point in the dataset. Given a query $q$, this procedure results in a vector of $N$ distances, where the $i^{\text{th}}$ position in the vector contains the distance $d(x_i,q)$. If $j$ is the index of the smallest element in the vector, then $x_j$ is the nearest neighbor to the query. Now suppose that we are given a vector $\mathbf{s}$ of $N$ kernel evaluations rather than explicit distances. Here, the $i^{\text{th}}$ component of $\mathbf{s}$ is $s_i = k(x_i,q)$, where $k(\cdot,\cdot)$ is a radial kernel. Radial kernels are nearly 1 when $d(x_i,q)$ is small and decrease to 0 as $d(x_i,q)$ increases. Since $k(x_i,q)$ is a monotone decreasing function with respect to $d(x_i,q)$, the vector of kernel values is also sufficient to perform near neighbor search. If $s_j$ is the largest component of $\mathbf{s}$, then $x_j$ is the nearest neighbor to the query. The main idea of our algorithm is to apply compressed sensing techniques to $\mathbf{s}$.

The main result from compressed sensing is that a sparse vector $\mathbf{s}$ can be recovered from a sub-linear memory sketch of its components. If we assume that $\mathbf{s}$ is $v$-sparse (contains only $v$ elements that are large), then we can recover $\mathbf{s}$ from $O(v \log N/v)$ random linear combinations of the entries of $\mathbf{s}$. The key insight is that each measurement is a weighted kernel density estimate (KDE) over the dataset. Using a small collection of KDE sums, we can identify the near neighbors of the query. If we choose the coefficients to be $\{1,0\}$, then each measurement is an unweighted KDE over a partition of the dataset. While it requires $N$ memory to compute the exact KDE, recent results~\cite{coleman2020race} show that the KDE may be approximated by an online sketch in space that is constant with respect to $N$. While larger sketches improve the quality of the approximation, the memory does not grow when elements are added to the dataset. Thus, each of the $O(v \log N/v)$ measurements can be approximated using constant memory in the streaming setting.

\section{Theory}
\label{sec:theory}
Due to space constraints, we omit proofs and corner cases. For a thorough presentation that includes proofs, see the supplementary material. 

\subsection{Estimation of Compressed Sensing Measurements}
\label{sec:estimateCSmeasurements}
To bound the error of the approximation for our compressed sensing measurements, we bound the variance of the RACE estimator using standard inequalities. 



\begin{theorem}
\label{thm:ACELinearCombos}
Given a dataset $\mathcal{D}$, $K$ independent LSH functions $l(\cdot)$ and any choice of constants $r_i \in \mathbb{R}$, RACE can estimate a linear combination of $\mathbf{s}_i(q) = p(x_i,q)^K$ with the following variance bound. 
\begin{equation}
         \E[A[L(q)]] = \sum_{x_i \in \mathcal{D}} r_i p(x_i,q)^K
\end{equation}
\begin{equation}
\var(A[l(q)]) \leq |\tilde{\mathbf{s}}(q)|^2_1\\
\end{equation}
where $L(\cdot)$ is formed by concatenating the $K$ copies of $l(\cdot)$ and $\tilde{\mathbf{s}}_i(q) = \sqrt{\mathbf{s}_i(q)}$. 
\end{theorem}
Let $\mathbf{y}\in \mathbb{R}^M$ be the $M$ compressed sensing measurements of the KDE vector $\mathbf{s}(q)$. A direct corollary of Theorem~\ref{thm:ACELinearCombos} is that by setting the coefficients correctly, we can obtain unbiased estimators of each measurement with bounded variance. Using the median-of-means (MoM) technique, we can obtain an arbitrarily close estimate of each compressed sensing measurement. To ensure that all $M$ measurements obey this bound with probability $1-\delta$, we also apply the probability union bound. Note that the multiplicative $M$ factor comes from the fact that we are using ACE to estimate $M$ different measurements. 

\begin{theorem}
\label{thm:AllACEErrorBound}
Given any $\epsilon > 0$ and
$ O\Big(M \frac{|\tilde{\mathbf{s}}(q)|^2_1}{\epsilon^2}\log\Big(\frac{M}{\delta}\Big)\Big)$
independent ACE repetitions, for any query $q$, we have the following bound for each of the $M$ measurements with probability $1-\delta$
\begin{equation}
y_i(q) - \epsilon \leq \hat{y}_i(q) \leq y_i(q) + \epsilon\\
\end{equation}
\end{theorem}


Therefore, by repeating ACE estimators (RACE), we can obtain low-variance estimates of the compressed sensing measurements of $\mathbf{s}(q)$. The exact number of measurements $M$ depends on both $\Phi$ and the dataset, but $M < O(N)$. 

\subsection{Query-Dependent Sparsity Conditions}
\label{sec:sparsityconditions}
For our compressed sensing measurements to be useful, $\mathbf{s}(q)$ needs to be \textit{sparse} with a bound on $|\mathbf{s}(q)|_1$\cite{donoho2006compressed}. We also require a bound on $|\tilde{\mathbf{s}}(q)|_1$ to avoid a memory blow-up in Theorem~\ref{thm:AllACEErrorBound}. If we simply assume a bound on $|\tilde{\mathbf{s}}(q)|_1$, it is straightforward to show that the sketch requires sub-linear memory. See the supplementary materials for details. To characterize the type of queries that are appropriate for our algorithm, we connect sparsity with the idea of near-neighbor stability~\cite{beyer1999nearest}, a well-established notion of query difficulty.

Given any vector $\mathbf{s}(q)$ with elements between 0 and 1, we can tune $K$ to make $\mathbf{s}(q)$ sparse and obtain the required bounds. However, increasing $K$ also increases the memory because we require increasingly more precise estimates to differentiate between $\mathbf{s}_{v}$ and $\mathbf{s}_{v+1}$. Therefore, we want $K$ to be just large enough. The largest value of $K$ is required when all points in the dataset other than the $v$ nearest neighbors are equidistant to the query ($|\mathbf{s}|_1 = O(N)$). To choose $K$ appropriately, we begin by defining two data-dependent values $\Delta$ and $B$ to characterize this situation. Suppose that $x_v$ and $x_{v+1}$ are the $v^{\text{th}}$ and $(v+1)^{\text{th}}$ nearest neighbors, respectively. Let $\Delta$ be defined as $\Delta = \frac{p(x_{v+1},q)}{p(x_v,q)}$ and $B = \sum_{i = v+1}^N \frac{\tilde{s}_i}{\tilde{s}_{v+1}}$. $\Delta$ measures the stability (Definition~\ref{def:unstableNN}) of the query and is a measure of the gap between the near-neighbors and the rest of the dataset. If $\Delta \approx 1$, then $x_v$ and $x_{v+1}$ are very difficult to separate and the query is \textit{unstable}. $B$ measures the sparsity of $\mathbf{s}$. If $B$ is $O(N)$, then every element of $\mathbf{s}$ is nonzero (Figure~\ref{fig:sparsity}). We are now ready to present our results for $K$ in terms of $B$ and $\Delta$. 
\begin{theorem}
\label{thm:ChooseKtoBoundS}
Given a query $q$ and query-dependent parameters $B$ and $\Delta$, if 
$K = \ceil[\Big]{2\frac{\log B}{\log \frac{1}{\Delta}}}$
then
$ p(x_v,q)^K \geq \sum_{i=v+1}^{N} p(x_i,q)^K$
and we have the bounds 
$ |\mathbf{s}(q)|_1 \leq v+1$ and $ |\tilde{\mathbf{s}}(q)|_1 \leq v+1$
\end{theorem}

\begin{figure*}[t]
\centering
\mbox{\hspace{-0.2in}
\includegraphics[height=1.8in]{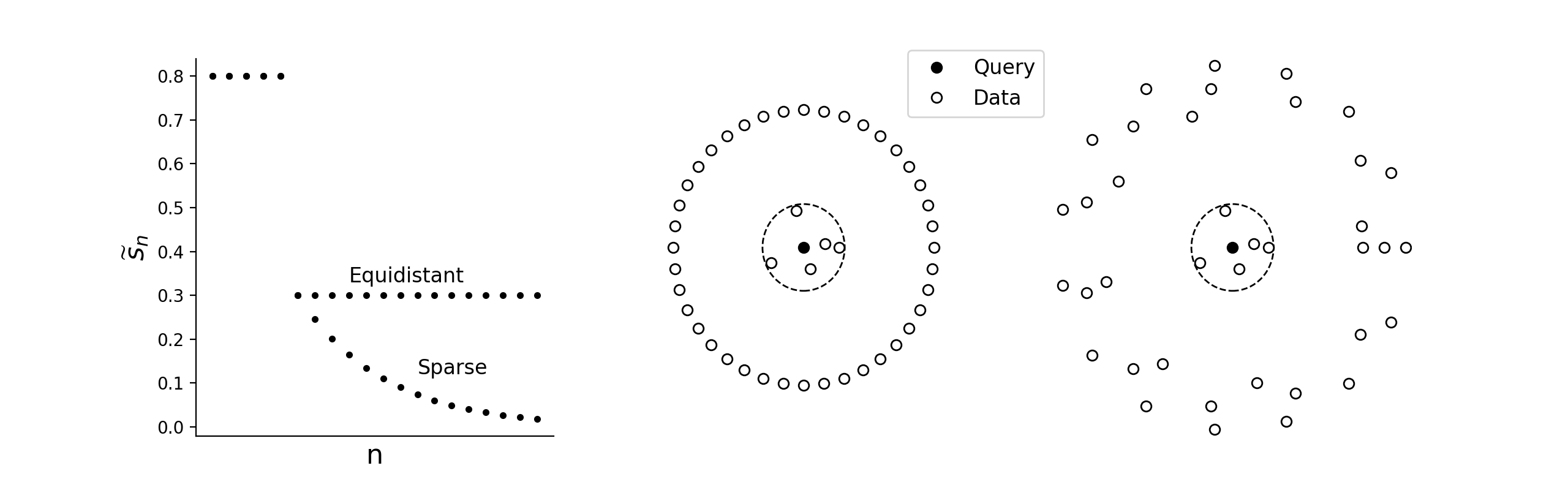}}
\vspace{-0.1in}
\caption{Geometric interpretation of $B$ and $\Delta$. $\Delta$ characterizes the gap between the $v$ nearest neighbors, while $B$ characterizes whether $\mathbf{s}$ is sparse. The worst-case situation occurs when all points are equidistant to the query (center). However, if $\mathbf{s}$ is already sparse, then far fewer points in the dataset are near the query (right).}
\label{fig:sparsity}
\end{figure*}

\begin{figure*}[t]
\centering
\mbox{\hspace{-0.2in}
\includegraphics[height=2.34in]{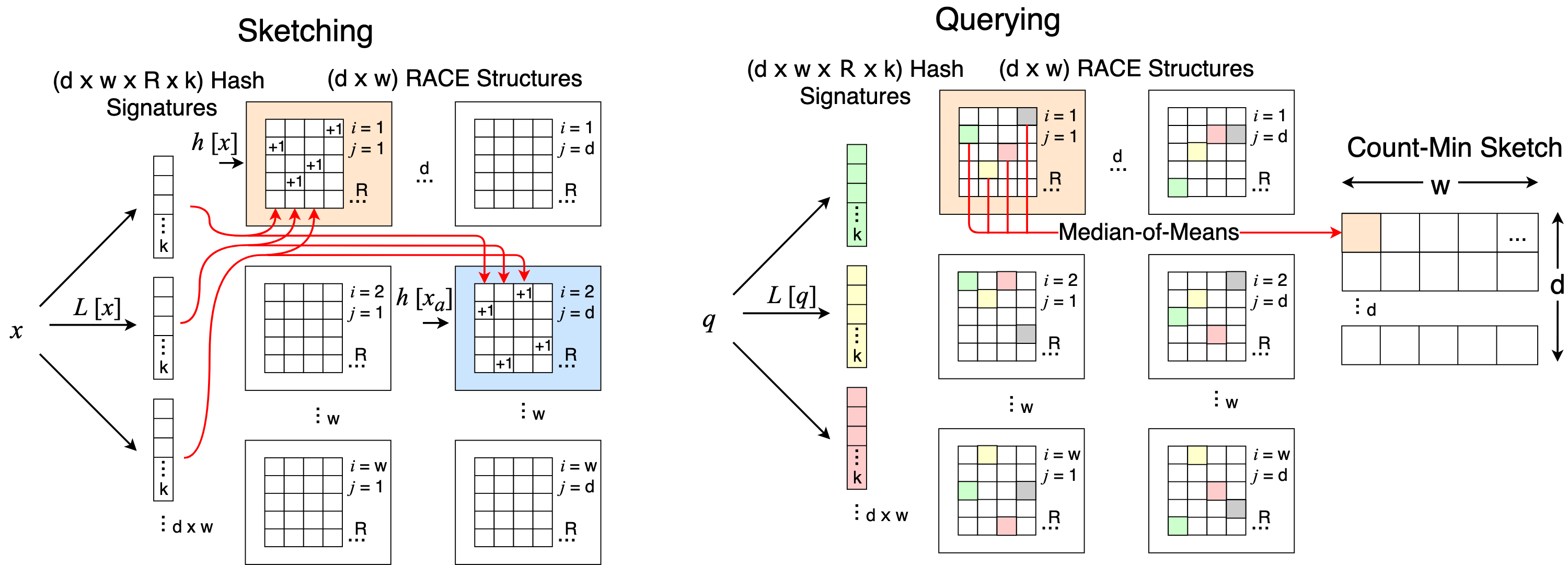}}
\vspace{-0.1in}
\caption{Implementation of sketching (Algorithm~\ref{alg:sketch}) and querying (Algorithm~\ref{alg:query}) using RACE data structures. During sketching, we compute $d\times w \times R \times k$ hash values for each $x \in \mathcal{D}$ and update the RACEs selected using $h(\cdot)$. During querying, we compute the hash values of the query $q$ and estimate the CMS measurements.}
\label{fig:algodiagram}
\end{figure*}

In practice, this assumption is unrealistically pessimistic because $\mathbf{s}(q)$ is often sufficiently sparse without any intervention using $K$. However, Theorem~\ref{thm:ChooseKtoBoundS} always allows us to choose $K$ so that $|\mathbf{s}(q)|_1$ is bounded by a constant. 



\subsection{Reduce Near-Neighbor to Compressed Recovery}
\label{sec:reduceNNtoCSforCMS}
We can apply Theorem~\ref{thm:AllACEErrorBound} to estimate each of the $M$ CMS measurements, which we call $\widehat{\CMS}$. We want to recover an estimate $\hat{\mathbf{s}}$ of $\mathbf{s}$ from our approximate compressed sensing measurements $\widehat{\CMS}$. Since the error $\epsilon_E$ in our approximation simply adds to the CMS recovery error $ \epsilon_C$ from Theorem~\ref{thm:CMSPointwiseGuarantee}, we can recover the values of $\mathbf{s}(q)$ by choosing appropriate values for $\epsilon_C$ and $\epsilon_E$. 
\begin{theorem}
\label{thm:PointwiseReconstructionBoundCMS}
We require 
$$ O\left(\frac{|\tilde{\mathbf{s}}(q)|_1^2 |\mathbf{s}(q)|_1}{\epsilon^3} \log\left(\frac{|\mathbf{s}(q)|_1}{\epsilon \delta}\log\left(\frac{N}{\delta}\right)\right)\log\left(\frac{N}{\delta}\right)\right)$$
ACE estimates to recover $\hat{\mathbf{s}}(q)$ with probability $1-\delta$ such that 
\begin{equation}
s_i(q) - \frac{\epsilon}{2} \leq \hat{s_i}(q) \leq s_i(q) + \frac{\epsilon}{2}
\end{equation}
\end{theorem}

If $\mathbf{s}$ is sparse, then this result can be used to identify the top $v$ elements of $\mathbf{s}$ by setting $\epsilon = s_{v} - s_{v+1} = p_v^K - p_{v+1}^K$. These elements correspond to the largest kernel evaluations and therefore the nearest neighbors. For the equidistant case, we substitute the value of $K$ from Theorem~\ref{thm:ChooseKtoBoundS} into the expression in Theorem~\ref{thm:PointwiseReconstructionBoundCMS} to obtain our final results. Our main theorem is a simplified result that relates the size of the RACE sketch with the query-dependent parameters $\Delta$ and $p_v$. The full derivation, including the dependence on $\delta$, is available in the supplementary materials. 

\begin{theorem}
\label{thm:RACEforVNN_CMS}
It is possible to construct a sketch that solves the exact $v$-nearest neighbor problem with probability $1-\delta$ using $ O\left(N^b \log^3\left(N\right)\right)$
bits, where 
$$ b = \frac{6 |\log {p_v}| + 2\log r}{\log{\frac{1}{\Delta}}}$$
Here, $r$ is the range of the LSH function, and $p_v$ is the collision probability of the $v^{\text{th}}$ nearest neighbor with the query.
\end{theorem}

\section{Experiments}

In this section, we rigorously evaluate our RACE-CMS sketch on friend recommendation tasks on social network graphs, similar to the ones described in~\cite{sharma2017hashes}. Our goal is to compare and contrast the practical compression-accuracy tradeoff of RACE with streaming baselines. We use the Google Plus social network dataset, obtained from~\cite{leskovec2012learning}, and the Twitter and Slashdot graphs from~\cite{snapnets} to evaluate our algorithm. Google Plus is a directed graph of 107,614 Google Plus users, where each element in the dataset is an adjacency list of connections to other users. The uncompressed dataset size is 121 MB when stored in a sparse format as the smallest possible unsigned integral type. The other datasets are structured the same way, with similar sizes. Additional statistics are displayed in Table \ref{table:1}. These characteristics are typical for large scale graphs, where the data is high dimensional and sparse. Note that the low mean similarity between elements indirectly implies that $\mathbf{s}(q)$ is sparse.

\begin{table}[t]
\caption{Dataset Statistics}
\centering
\begin{tabular}{c|c| c| c| c }
 \hline
 Dataset & Nodes & Nonzeros & \makecell{Mean\\Edges} & \makecell{Mean\\Similarity} \\ 
 \hline
 Google+ & 108k & 13.6M & 127 & 0.002 \\ 
 \hline
 Twitter & 81.4k & 1.8M & 22 & 2.2e-4 \\ 
 \hline
 Slashdot & 82.2k & 1.1M & 13 & 1.4e-5 \\ 
 \hline
\end{tabular}
\label{table:1}
\vspace{-0.1in}
\end{table}

\begin{figure*}[t]
\centering
\mbox{\hspace{-0.2in}
\includegraphics[width=2.4in]{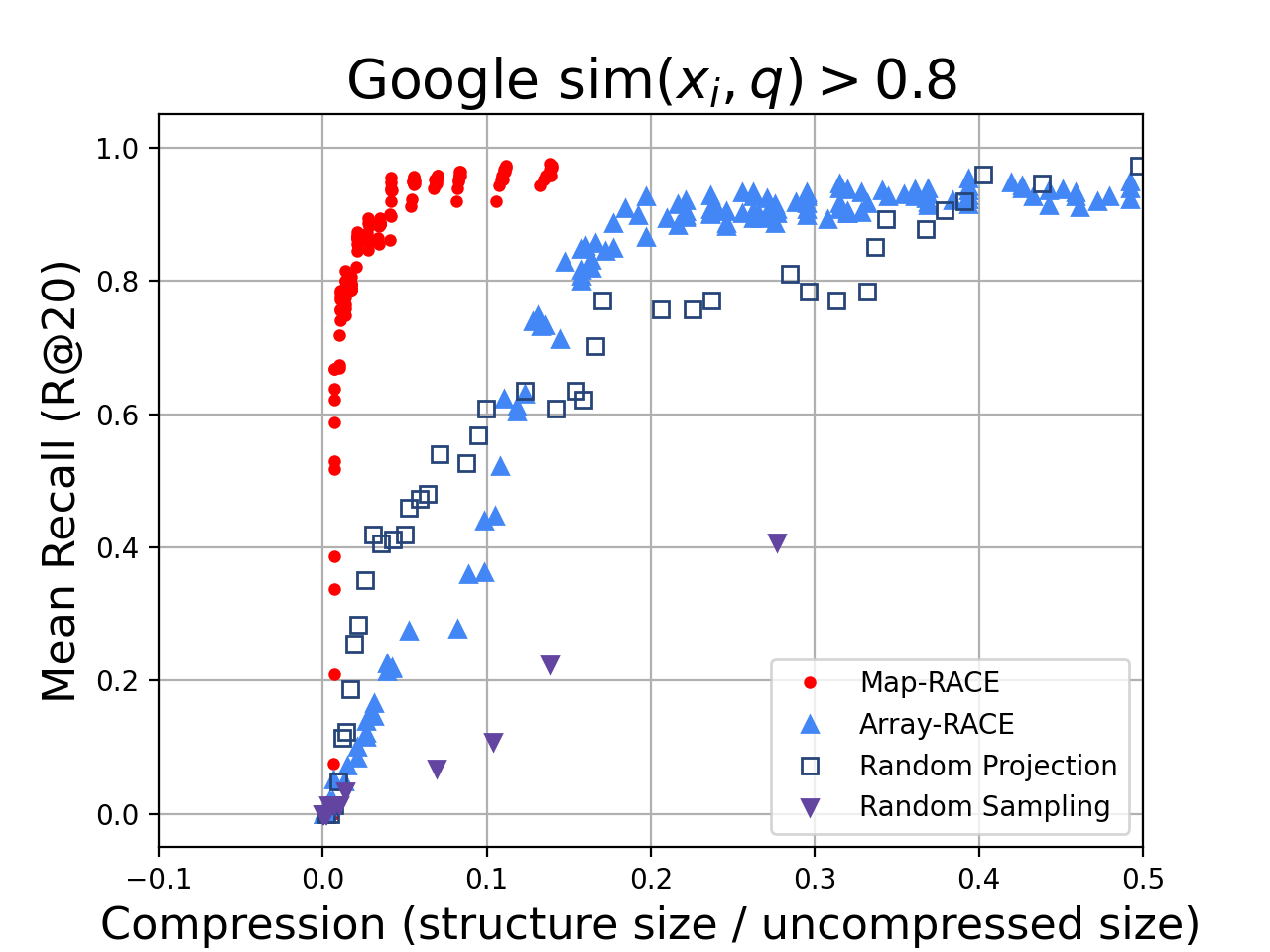}
\includegraphics[width=2.4in]{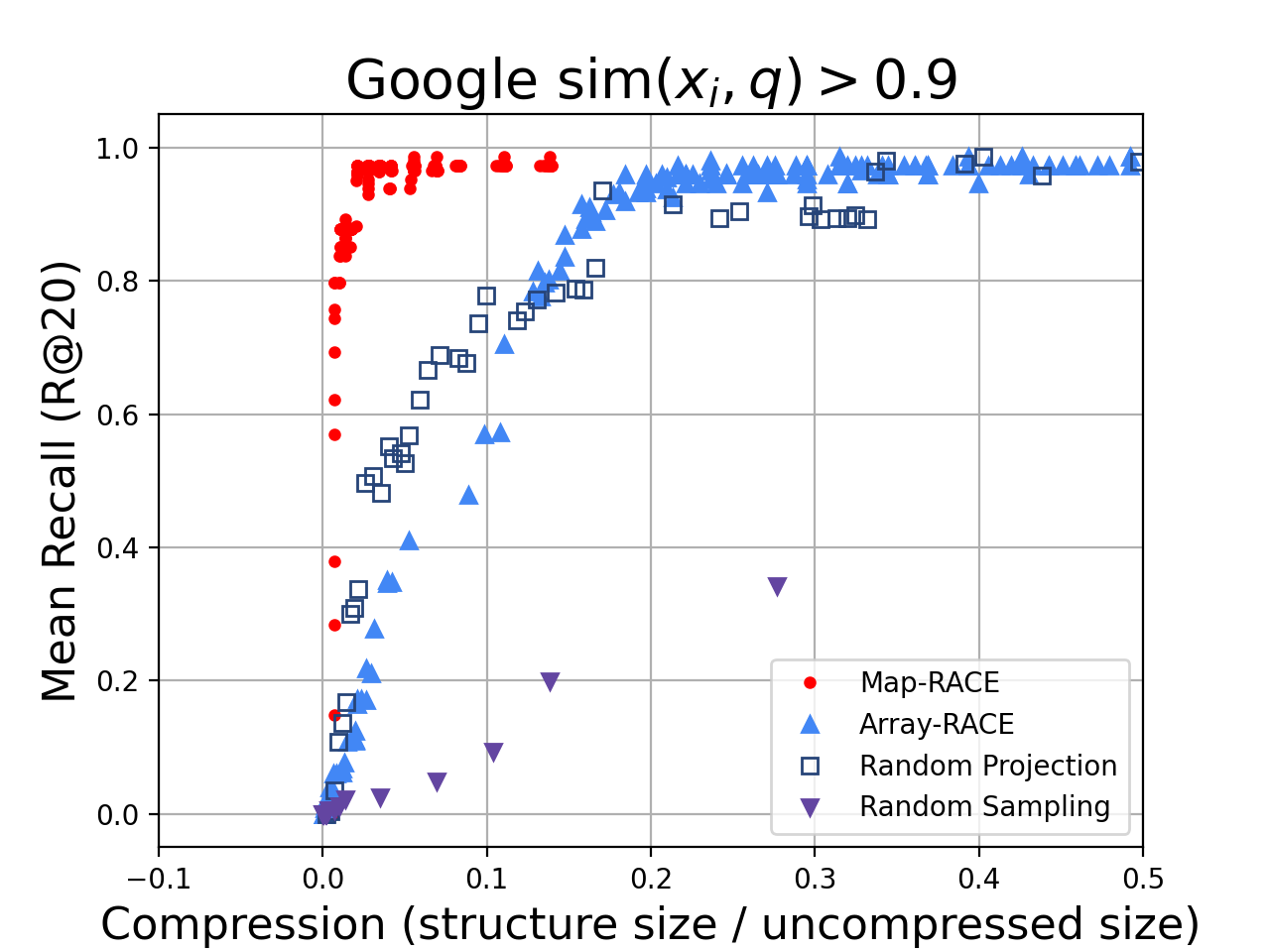}
}
\mbox{\hspace{-0.2in}
\includegraphics[width=2.4in]{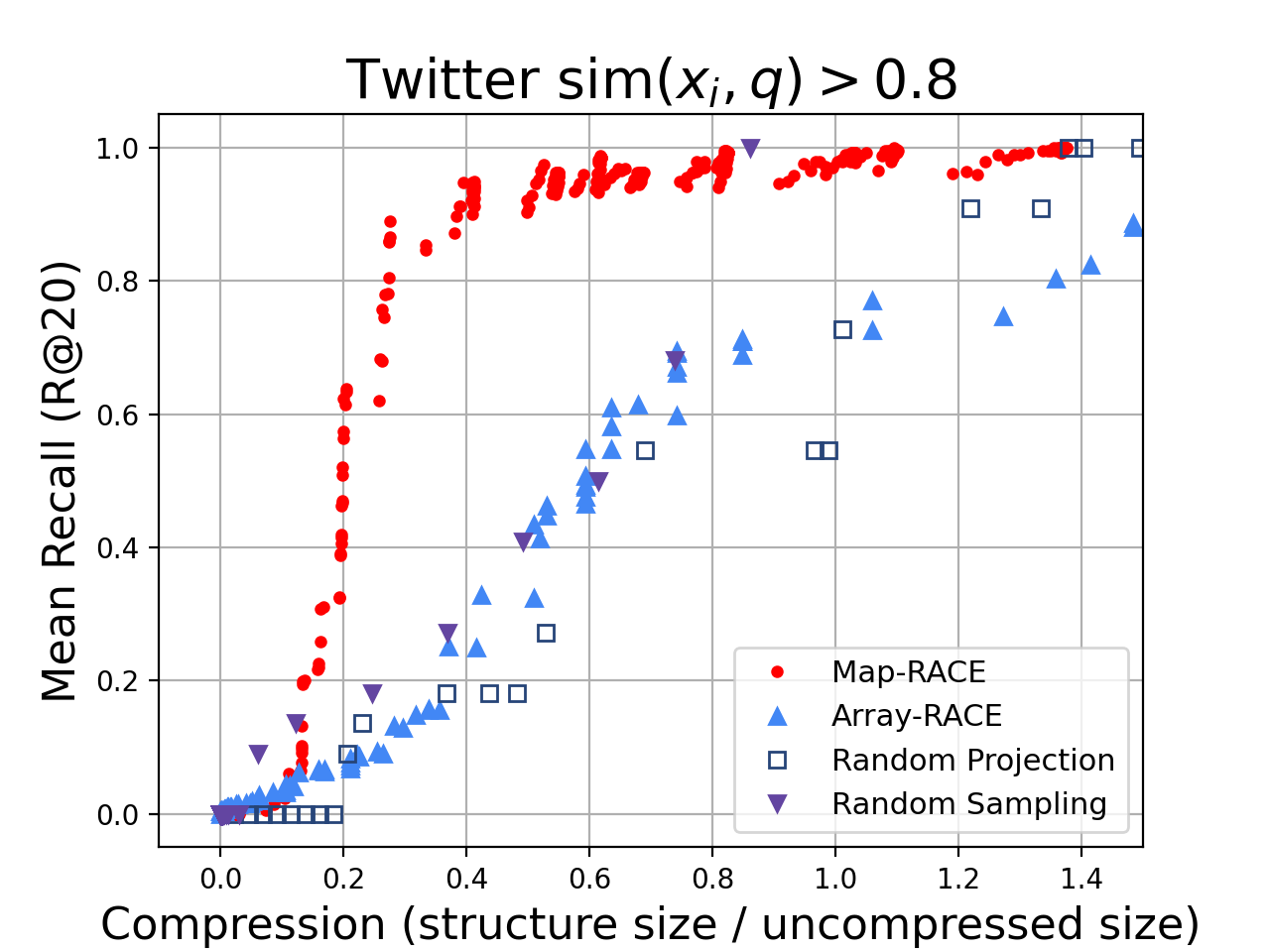}
\includegraphics[width=2.4in]{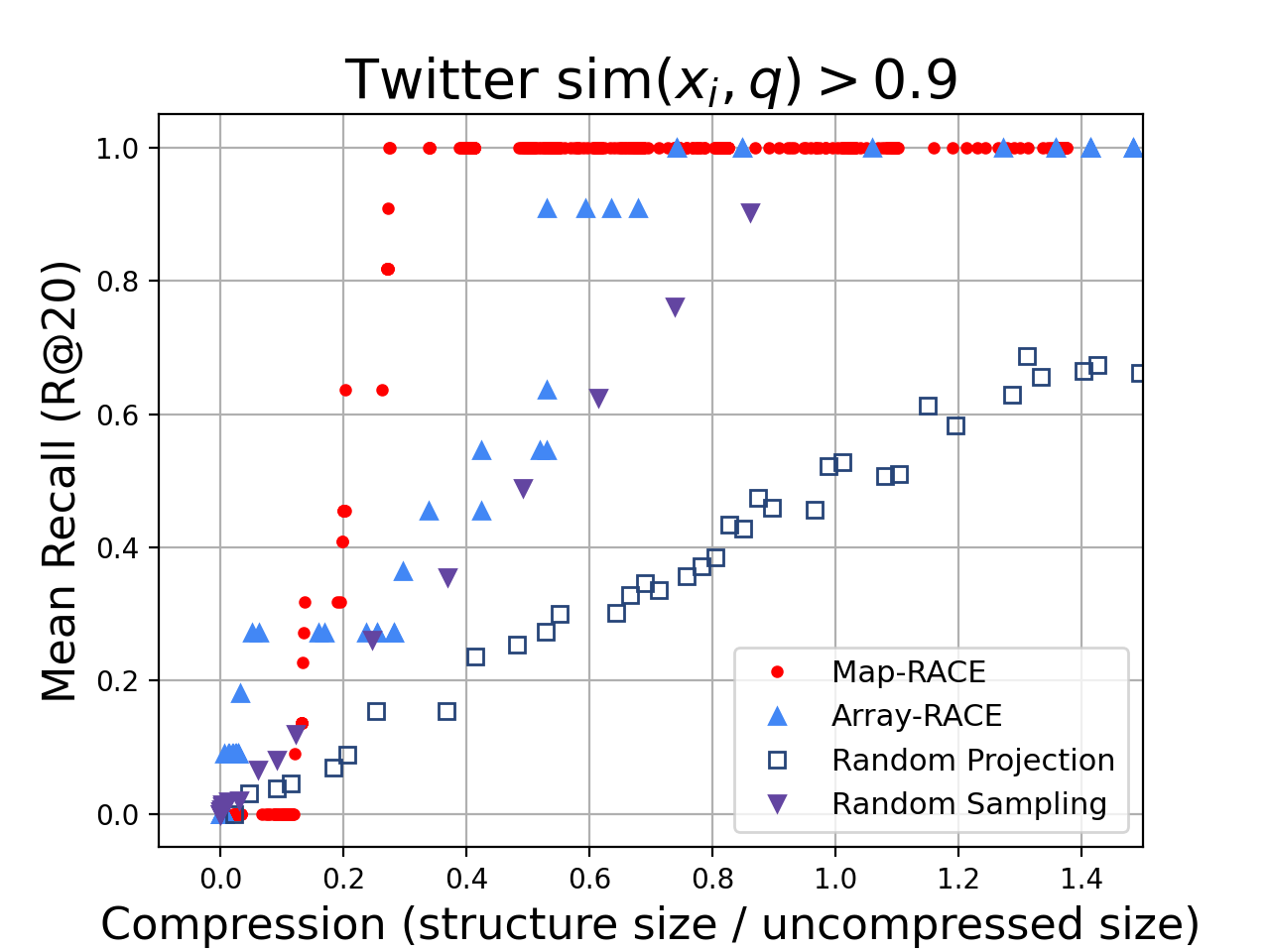}
}
\mbox{\hspace{-0.2in}
\includegraphics[width=2.4in]{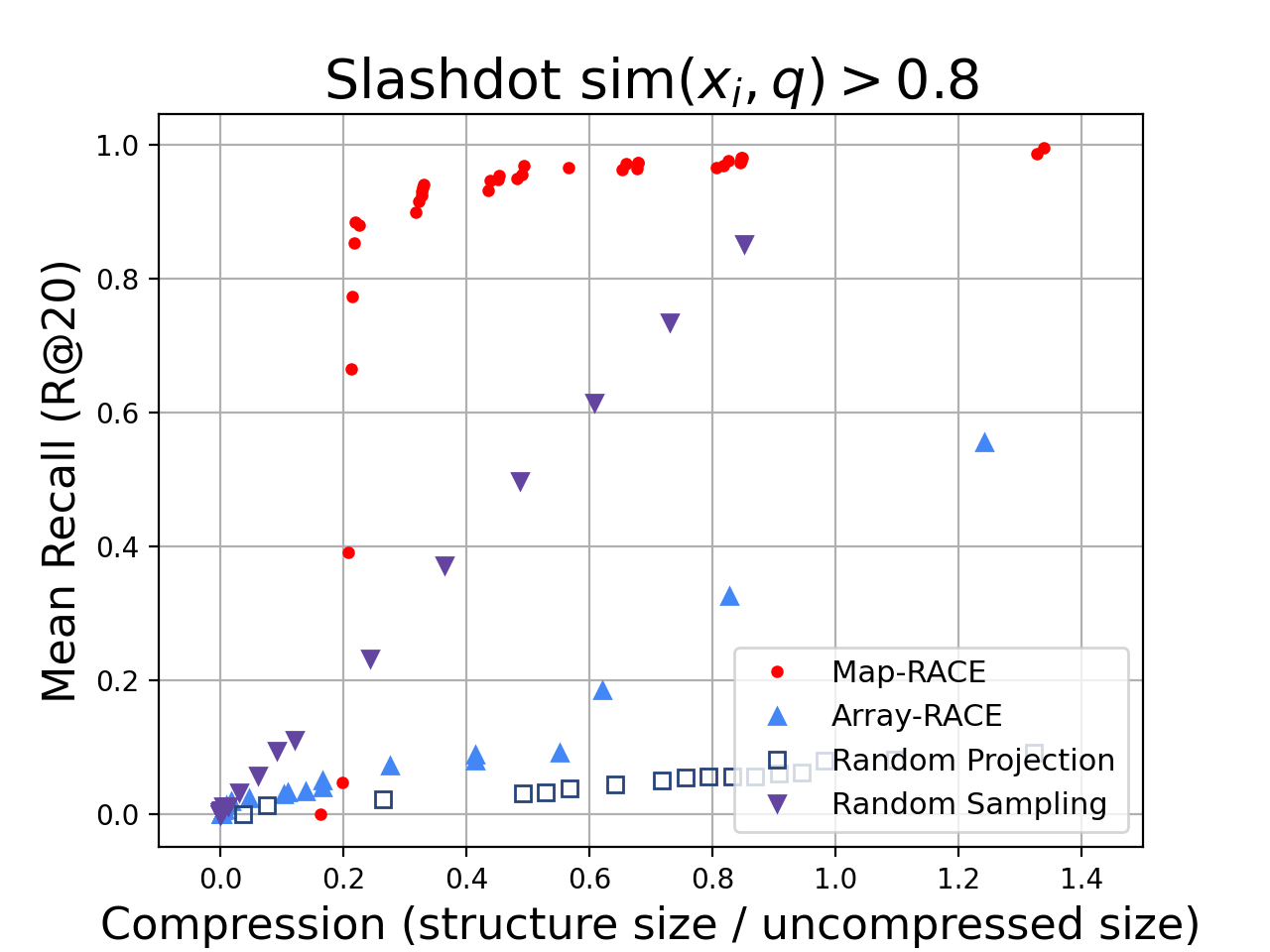}
\includegraphics[width=2.4in]{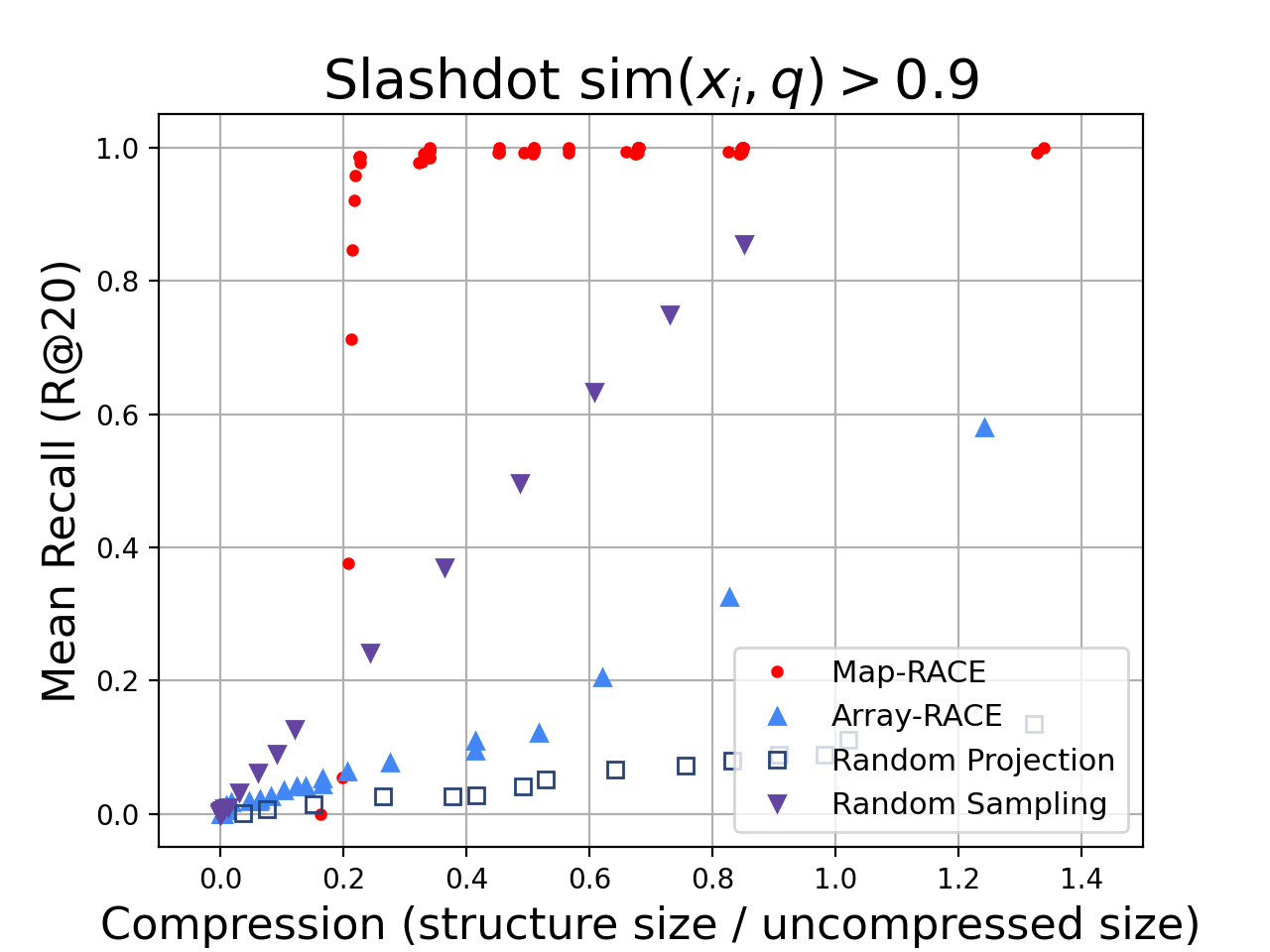}
}
\vspace{-0.1in}
\caption{Average recall vs compressed dataset size. The dataset size is expressed as the inverse compression ratio, or the ratio of the compressed size to the uncompressed size. Recall is reported as the average recall of neighbors with Jaccard similarity $\similarity(x,q) \geq 0.8$ (left) and $0.9$ (right) over the set of queries. Higher is better. We report the recall of nodes with similarity greater than or equal to $0.8$ and $0.9$ for the top 20 search results of the query. Results are averaged over $>500$ queries. }
\label{fig:accuracyplot}
\end{figure*}

\subsection{Implementation}
We use the RACE-CMS sketch that was presented in Section~\ref{sec:theory}. However, we slightly deviate from the algorithm described in Algorithm~\ref{alg:sketch} in our implementation by rehashing the $K$ LSH hash values to a range $r$ using a universal hash function. Our algorithm is characterized by the hyperparameters $K, d, w, R$ and $r$ and by the hash functions $l(\cdot)$ and $h(\cdot)$. Here, $l(\cdot)$ is MinHash, an LSH function for the Jaccard distance. We use MurmurHash for $h(\cdot)$, the universal hash function in the CMS. For all experiments, we vary $K$, $d, w$ and $R$ to trade off memory for performance. We present the operating points on the Pareto frontier for all algorithms. Typical values of $d$ are between 2 and 5, $w$ between 100 and 1000, and $R$ between 2 and 8. We varied the range $r$ between 100 and 1000 and used $K \in \{1,2\}$.

We implemented RACE-CMS in C++ with the following considerations. First, we do not store the RACE counters as full 32-bit integers. The count values tend to be small because the CMS only assigns each data point to $d$ cells out of $dw$ total cells, and each each cell further divides the counts into the RACE arrays. In our evaluation, we used 16-bit short integers, although more aggressive memory optimizations are likely possible. For example, we found that all counts were less than 32 in our Google Plus experiments, suggesting that 8-bit integer arrays are sufficient. The second optimization comes from our observation that many count values are zero. By storing the RACE sketches as sparse arrays or maps, we do not have to store the zero counts. We present results for the situation where we store dense arrays of counts (Array-RACE) and where we store RACE as a sparse array (Map-RACE). An implementation diagram is shown in Figure~\ref{fig:algodiagram}. 


\subsection{Baselines}
We compare our method with dimensionality reduction and random sampling followed by exact near-neighbor search. We reduce the size of the dataset until a given compression ratio is achieved and then find the nearest neighbors with the Euclidean distance. We compare against all methods that can operate in the strict one-pass streaming environment~\cite{fiat1998online}, which is required in many high-speed applications. We considered a comparison with product quantization using FAISS~\cite{johnson2017billion} but we encountered issues due to the dimensionality ($> 100k$) of our graph data, which agrees with previous evaluations of FAISS on high-dimensional data~\cite{wang2018randomized}. Samples are represented using 32-bit integer node IDs and are stored in sparse format, since the graph vectors tend to have many zeros. Projections are stored as dense arrays of single-precision (32-bit) floating point numbers.


\textbf{Random Projections:} We use sparse random projections~\cite{achlioptas2003database} and the Johnson-Lindenstrauss lemma to reduce the dimensionality of the dataset. This is the best known streaming method that is also practical. \\
\textbf{Random Sampling:} With random sampling, we reduce the original dataset to the desired size by selecting a random subset of elements of the dataset. Given a query, we perform exact nearest neighbor search on the random samples. 

\subsection{Experimental Setup}
We computed the ground truth Jaccard similarities and nearest neighbors for each vector in the dataset. We are primarily interested in queries for which high similarity neighbors exist in the dataset due to the constraints of the friend recommendation problem. This is also consistent with the near-neighbor problem statements in Section~\ref{sec:problemstatements}, which assume the existence of a near-neighbor. We return the 20 nearest neighbors and report the recall of points with similarity greater than 0.8 and 0.9 to the query. To confirm that the sketch is not simply memorizing our queries, we remove the query from the dataset before creating the sketch.

For random projections, we performed a sweep of the number of random projections from 5 to 500. Random sampling was performed by decimating the dataset (without replacement) so that the sampled dataset had the desired size. 




\subsection{Results}
Figure \ref{fig:accuracyplot} shows the mean recall of ground-truth neighbors for the RACE-CMS sketch. Array-RACE and Map-RACE are both implementations of our method, but with a different underlying data structure used to represent the RACE sketch. We obtained good recall ($>0.85$) on the set of queries with high-quality neighbors ($\mathrm{sim}(x,q) > 0.9$) even for an extreme 20x compression ratio on Google Plus and 5x compression ratios on the other datasets. Since many entries in the array are zero, we find that Map-RACE outperforms Array-RACE by a sizeable margin.

It is evident that RACE performs best for high similarity search. This is due to increased sparsity of $\mathbf{s}(q)$ (any two random users are unlikely to share a friend and hence have similarity zero) and higher $p(x_v,q)$. In the recommender system setting, we usually wish to recommend nodes with very high similarity. If we require the algorithm to recover neighbors on the Google Plus graph with similarity measure greater than 0.9 with an expected recall of 80\% or higher, our algorithm requires only 5\% of the space of the original dataset (6 MB) while random projections require 60 MB (50\%) and random sampling requires nearly the entire dataset. For neighbors with lower similarity (0.8), our method requires roughly one quarter of the memory needed by random projections.

\section{Conclusion}
We have presented RACE-CMS, the first sub-linear memory algorithm for near-neighbor search. Our analysis connects the stability of a near-neighbor search problem with the memory required to provide an accurate solution. Additionally, our core idea of using LSH to estimate compressed sensing measurements creates a sketch that can encode structural information and can process data not seen during the sketching process. 

We supported our theoretical findings with experimental results. In practical test settings, RACE-CMS outperformed existing methods for low-memory near-neighbor search by a factor of 10. We expect that RACE-CMS will enable large-scale similarity search for a variety of applications and will find utility in situations where memory and communication are limiting factors.

\section*{Acknowledgements}
This work was supported by National Science Foundation IIS-1652131, BIGDATA-1838177, RI-1718478, AFOSR- YIP FA9550-18-1-0152, Amazon Research Award, and the ONR BRC grant on Randomized Numerical Linear Algebra. 


\bibliography{main}
\bibliographystyle{icml2020}

\clearpage

\setcounter{theorem}{2}

\section{Supplementary Materials}

We obtain our results by combining recent advances in locality-sensitive hashing (LSH)-based estimation with standard compressed sensing techniques. This section contains a high-level overview of our strategy to solve the nearest-neighbor problem.

\paragraph{LSH-based kernel estimators:} The array-of-counts estimator (ACE) is an unbiased estimator for kernel functions. Our first step is to use ACE to estimate arbitrary linear combinations of kernels. We get sharp estimates of these linear combinations by averaging over multiple ACEs. We call this structure a RACE because it consists of \textit{repeated} ACEs. The number of repetitions needed for a good estimate does not depend on $N$, the dataset size. Once we have sharp estimates of the measurements, we apply standard compressed sensing techniques.

\paragraph{Compressed sensing:} A central result of compressed sensing is that a $v$ sparse vector of length $N$ can be recovered from $O(v \log N/v)$ linear combinations of its elements. The coefficients of the linear combination are defined by the measurement matrix. In this context, our measurements are of the vector $\mathbf{s}(q) \in \mathbb{R}^N$, where the $i^{\text{th}}$ component of $\mathbf{s}(q)$ is the kernel evaluation $k(x_i,q)$. Since we are using LSH kernels, $k(x_i,q)$ is the LSH collision probability of $q$ and $x_i \in \mathcal{D}$. That is, $\mathbf{s}_i(q) = p(x_i,q)$. We will use the terms LSH kernel value and collision probability interchangeably. 

We use one RACE structure to estimate each compressed sensing measurement. Each RACE gets a different set of linear combination coefficients, and we choose the coefficients so that they describe a valid measurement matrix. By applying sparse recovery to the set of estimated measurements, we can approximate the kernel evaluations (LSH collision probabilities) between $q$ and each element in the dataset. As explained in the ``Intuition'' section of the main text, these kernel evaluations are sufficient to perform near neighbor search. Assuming sparsity, the sketch is sublinear because each RACE requires a constant amount of memory and we only need to use $O(v \log N/v )$ RACEs. 

\paragraph{Query-dependent guarantees:} To find neighbors for a query ($q$), we recover the kernel values ($\mathbf{s}(q)$) and return the indices with the largest values as the identities of the near-neighbors. This process will not succeed if $\mathbf{s}(q)$ is not {\em sparse}. Sparsity is the reason for our query-dependent assumptions. To find the nearest indices, we need $\mathbf{s}(q)$ to have few large elements. The geometric interpretation is that most elements in the dataset are not near-neighbors of $q$. We show that well-established notions of near-neighbor stability~\cite{beyer1999nearest} are equivalent to weak sparsity conditions on $\mathbf{s}(q)$, allowing us to express our algorithm in terms of near neighbor stability. This result connects sparsity - a compressed sensing idea - with the difficulty of the near-neighbor search problem. We can analyze a large class of geometric data assumptions by interpreting them as sparsity conditions. 

If the dataset already satisfies our sparsity condition, then we proceed directly to recovery. If not, we can force $\mathbf{s}(q)$ to be sparse by raising the kernel function $k(x_i,q)$ to a power $K$. This modification decreases the bandwidth of the kernel, letting us locate near-neighbors at a finer resolution. RACE can accommodate this idea by using standard methods for amplifying a LSH family. Specifically, we construct the LSH function from $K$ independent realizations of an LSH family. The result is a new LSH function with the collision probability $p(x,q)^K$. However, there is a price - the size of each ACE repetition grows larger. 

\paragraph{Reduce near-neighbor to compressed sensing recovery:}Using compressed sensing, we can estimate the kernel values within an $\epsilon$ additive tolerance. To solve the near-neighbor problem, we make $\epsilon$ small enough to distinguish between near-neighbors and the rest of the dataset. The value of $\epsilon$ depends on $K$. Increasing $K$ makes $\mathbf{s}(q)$ sparse but also increases the amount of storage required for the sketch. Therefore, we want $K$ to be \textit{just large enough}. By balancing the sparsity requirement with the memory, we introduce a query-dependent multiplicative $O(N^b)$ factor for the sketch size. This term is sub-linear ($b < 1$) when $\mathbf{s}(q)$ is sufficiently sparse or, equivalently, when $q$ is a stable query. Our sketch requires $O(N^b \log^3 (N))$ bits, where $b$ depends on the query stability.

\section{Theory}
In this section, we provide a detailed explanation of the theory with complete proofs.

\subsection{Estimation of Compressed Sensing Measurements}
\label{sec:estimateCSmeasurements}
In this section, our goal is to prove that the RACE algorithm can estimate the compressed sensing measurements of $\mathbf{s}(q)$, the vector of kernel evaluations. We begin by constructing a modified version of ACE that can estimate any linear combination of $\mathbf{s}(q)$ components. Then, we derive a variance bound on this estimator and apply the median of means technique. 

We can estimate the linear combination by incrementing the ACE array using the linear combination coefficients. Suppose we are given a sequence of linear combination coefficients $\{r_i\}_{i = 1}^N$. The original ACE estimator simply increments the array $A$ at index $L(x_i)$ by 1. We will use the notation $\I_i$ to refer to the indicator function $\I_{L(x_i) = L(q)}$. That is, $\I_i$ is 1 when the query collides with element $x_i$ from the dataset. For the original ACE algorithm, $A[L(q)] = \sum_{x_i\in\mathcal{D}} \I_i$. In our case, we increment $A[L(x_i)]$ by $r_i$ and therefore we have $A[L(q)] = \sum_{x_i\in\mathcal{D}} r_i \I_i$. The expectation of this estimator is the linear combination of LSH kernels (collision probabilities).

\begin{theorem}
\label{thm:ACELinearCombos}
Given a dataset $\mathcal{D}$, $K$ independent LSH functions $l(\cdot)$ and any choice of constants $r_i \in \mathbb{R}$, RACE can estimate a linear combination of $\mathbf{s}_i(q) = p(x_i,q)^K$ with the following variance bound. 
\begin{equation}
         \E[A[L(q)]] = \sum_{x_i \in \mathcal{D}} r_i p(x_i,q)^K
\end{equation}
\begin{equation}
\var(A[l(q)]) \leq |\tilde{\mathbf{s}}(q)|^2_1\\
\end{equation}
where $L(\cdot)$ is formed by concatenating the $K$ copies of $l(\cdot)$ and $\tilde{\mathbf{s}}_i(q) = \sqrt{\mathbf{s}_i(q)}$. 
\end{theorem}
\begin{proof}
For the sake of presentation, let $Z = A[L(q)]$.

\paragraph{Expectation:} The count in the array can be written as
$$Z = \sum_{x_i \in \mathcal{D}} r_i \I_i$$
By linearity of the expectation operator
$$\E[Z]\ = \sum_{x_i \in \mathcal{D}} r_i \E[\I_i]$$
$\E[\I_i]$ is simply the collision probability of $L$, thus 
$$\E[Z]\ = \sum_{x_i \in \mathcal{D}} r_i p(x_i,q)^K$$
\paragraph{Variance:} The variance is bounded above by the second moment. The second moment of this estimator can be written as 
$$\E[Z^2]\ = \sum_{x_i \in \mathcal{D}}\sum_{x_j \in \mathcal{D}} r_i r_j \E[\I_i \I_j]$$
Use the Cauchy-Schwarz inequality to bound $\E[\I_i \I_j] \leq \sqrt{\E[\I_i]}\sqrt{\E[\I_j]}$. Thus 
$$\E[Z^2]\ \leq \sum_{x_i\in \mathcal{D}} \sum_{x_j\in \mathcal{D}} r_i r_j \sqrt{p(x_i,q)^K}\sqrt{p(x_j,q)^K}$$
$$= \left(\sum_{x_i\in\mathcal{D}} r_i \sqrt{p(x_i,q)^K})\right)^2$$

For our analysis, we will assume that $r_i \in [-1,1]$. This is valid because we can always scale the compressed sensing matrix so that it is true. Then the bound becomes
$$\var(Z) \leq \left(\sum_{x_i\in\mathcal{D}} \sqrt{p(x_i,q)^K})\right)^2 = |\mathbf{\tilde{s}}(q)|_1^2$$

\end{proof}

Using Theorem~\ref{thm:ACELinearCombos} and the median-of-means (MoM) technique, we can obtain an arbitrarily close estimate of each compressed sensing measurement $y_i(q)$. Suppose we independently repeat the ACE estimator and compute the MoM estimate from the repetitions. Let $\hat{y}_i(q)$ be the MoM estimate of $y_i(q)$ computed from a set of independent ACE repetitions of $A[l(q)]$. Then we have a pointwise bound on the error for each $y_i(q)$. 

\begin{lemma}
\label{lem:OneACEErrorBound}
For any $\epsilon > 0$ and given
$$ O\Big(\frac{|\tilde{\mathbf{s}}(q)|^2_1}{\epsilon^2}\log\Big(\frac{1}{\delta}\Big)\Big) $$
independent ACE repetitions, we have the following bound for the MoM estimator 
\begin{equation}
y_i(q) - \epsilon \leq \hat{y}_i(q) \leq y_i(q) + \epsilon\\
\end{equation}
with probability $1-\delta$ for any query $q$.
\end{lemma}
\begin{proof}
For presentation, we will drop the index and write $\hat{y}_i(q)$ as $\hat{y}$ where the context is clear. We use a very common proof technique with the median-of-means estimator $\hat{y}$. With probability at least $1 - \delta$ and $n$ independent realizations of the random variable, we can estimate the mean with MoM so that 
$$ \mathrm{Pr}\left[|\hat{y} - y| \leq \sqrt{32\frac{\var(\hat{y})}{n} \log\left(\frac{1}{\delta}\right)}\right] \geq 1 - \delta $$
We can substitute the variance bound from Theorem~\ref{thm:ACELinearCombos} in for $\var(\hat{y})$ without changing the validity of the inequality. To have the lemma, we need $|\hat{y} - y|\leq \epsilon$. We will choose $n$ to be large enough that 
$$ \sqrt{32\frac{|\mathbf{\tilde{s}}(q)|_1^2}{n} \log\left(\frac{1}{\delta}\right)} \leq \epsilon$$
Therefore, we need $n$ ACE repetitions, where $n$ is
$$ n = 32 \frac{|\tilde{\mathbf{s}}(q)|^2_1}{\epsilon^2}\log\Big(\frac{1}{\delta}\Big) $$

\end{proof}

Lemma~\ref{lem:OneACEErrorBound} only works for one of the compressed sensing measurements. To ensure that all $M$ of the measurements obey this bound with probability $1-\delta$, we apply the probability union bound to get Theorem~\ref{thm:AllACEErrorBound}. Note that the multiplicative $M$ factor comes from the fact that we are using ACE to estimate $M$ different measurements. 

\begin{theorem}
\label{thm:AllACEErrorBound}
For any $\epsilon > 0$ and given
$ O\Big(M \frac{|\tilde{\mathbf{s}}(q)|^2_1}{\epsilon^2}\log\Big(\frac{M}{\delta}\Big)\Big)$
independent ACE repetitions, we have the following bound for each of the $M$ measurements with probability $1-\delta$ for any query $q$
\begin{equation}
y_i(q) - \epsilon \leq \hat{y}_i(q) \leq y_i(q) + \epsilon\\
\end{equation}
\end{theorem}
\begin{proof}
We want all measurements to succeed with probability $1 - \delta$. The probability union bound states that if $\delta_i$ is the failure probability for measurement $i$, then the overall failure probability is smaller than $\sum_{i = 1}^M \delta_i$. We would like this probability to be smaller than $\delta$, so we put $\delta_i = \frac{\delta}{M}$ for each RACE estimator. By Lemma~\ref{lem:OneACEErrorBound}, we need 
$$ 32\frac{|\tilde{\mathbf{s}}(q)|^2_1}{\epsilon^2}\log\Big(\frac{M}{\delta}\Big)$$
repetitions for each measurement. There are $M$ measurements, so we need 
$$ 32M\frac{|\tilde{\mathbf{s}}(q)|^2_1}{\epsilon^2}\log\Big(\frac{M}{\delta}\Big)$$
repetitions in total. 
\end{proof}


\subsection{Query-Dependent Sparsity Conditions}
\label{sec:sparsityconditions}

Before we can discuss compressed recovery of $\mathbf{s}(q)$, we need to limit our analysis to vectors $\mathbf{s}(q)$ that are sparse. In this section, we introduce a permissive way to bound the sparsity of $\mathbf{s}(q)$ for our analysis. Our bounds are forgiving in the sense that we assume as little underlying sparsity as possible - with stronger assumptions, you can get better bounds. We also connect sparsity with near-neighbor stability. We analyze these conditions in the context of compressed sensing and computational geometry.

We need bounds for $|\mathbf{s}(q)|_1$ and $|\tilde{\mathbf{s}}(q)|_1$. Our vector $\mathbf{s}(q)$ has three properties that make these bounds possible. First, the collision probabilities are bounded: $p(x_i,q) \in [0,1]$. Second, increasing $K$ causes each element of $\mathbf{s}(q)$ to decrease, since $s(q)_i = p(x_i,q)^K$. Third we may choose $K$ to be as large as necessary. Therefore, we can force $|\mathbf{s}(q)|_1$ to be arbitrarily small by choosing $K$ sufficiently large. However, each ACE estimator requires $O(r^K\log{N})$ memory where $r$ is the number of hash codes that $L$ can return. Therefore, we want $K$ to be \textit{just large enough} so that we do not increase the space too much. 

We will analyze sparsity under the \textit{equidistant} assumption. Under this assumption, all points other than the $v$ nearest neighbors are equidistant to the query. This is a relatively weak way to describe sparsity, but we still get an acceptable dependence of $K$ on $N$. Stronger assumptions require smaller $K$ and therefore less space. To choose $K$, we need a good way to characterize the sparsity of $\mathbf{s}(q)$. We begin by defining two query-dependent values $\Delta$ and $B$. $\Delta$ is related to the stability of the near-neighbor query and $B$ is related to sparsity.

\paragraph{$\Delta$-Stable Queries:}
We want a parameter that measures the difficulty of the query. For the $v$-nearest neighbor problem, let $x_v$ and $x_{v+1}$ be the $v^{\text{th}}$ and $(v+1)^{\text{th}}$ nearest neighbors, respectively. Using the same notation as before, let $\Delta$ be defined as 
\begin{equation}
    \Delta = \frac{p(x_{v+1},q)}{p(x_v,q)}
\end{equation}
$\Delta$ governs the stability of the nearest neighbor query. It is a measure of the gap between the near-neighbors and the rest of the dataset. If $\Delta = 1$, then the $v^{\text{th}}$ and $(v+1)^{\text{th}}$ neighbors are the same distance away. In this case, it is impossible to tell the difference between them. If $\Delta \approx 0$, then it means that neighbors $v+1$, $v+2$, ... are all very far away. Our definition of $\Delta$ is similar to the definition of an $\epsilon$-unstable query. In fact, we can express a $\Delta$-stable query as an $\epsilon$-unstable query by finding the distances that correspond to $p(x_{v+1},q)$ and $p(x_{v},q)$. This is possible because the collision probability is a monotone function of distance. 





\paragraph{$B$-Bounded Queries:} We want a flexible way to bound the sum:

$$|\mathbf{s}(q)|_1 = \sum_{x_i \in \mathcal{D}} p(x_i,q)^K$$

For convenience, we will suppose that the elements $x_i$ are sorted based on their distance from the query. This is not necessary - it just simplifies the presentation. When we write $p(x_i,q)$, we mean that $x_i$ is the $i^{\text{th}}$ near neighbor of the query. We will use the notation $p_i = p(x_i,q)$. Define a constant $B$ as
\begin{equation}
    B =\sum_{i = v+1}^N \frac{\tilde{s}_i}{\tilde{s}_{v+1}} = \sum_{i = v+1}^N \sqrt{\frac{p_i^K}{p_{v+1}^K}}
\end{equation}
$B$ is a query-dependent value that measures the sparsity of $\mathbf{s}$. It bounds the size of the tail entries of $\mathbf{s}$. A bound on $B$ implies a bound on $|\mathbf{s}|_1$ and $|\tilde{\mathbf{s}}|_1$. 
\begin{lemma}
\label{lem:vBBounds}
$|\mathbf{s}|_1 \leq |\tilde{\mathbf{s}}|_1$ and 
$|\tilde{\mathbf{s}}|_1 \leq v + B \sqrt{p_{v+1}^K}$
\end{lemma}
\begin{proof}

It is easy to see that $\mathbf{s}_i \leq \sqrt{\mathbf{s}_i}$ because $0 \leq \mathbf{s}_i \leq 1$. 
For the second inequality, break the summation for $|\tilde{\mathbf{s}}|_1$ into two components: 
$$ |\tilde{\mathbf{s}}|_1 = \sum_{i = 1}^v \sqrt{p_i^K} + \sum_{i = v+1}^N \sqrt{p_i^K}$$
The first term corresponds to the nearest $v$ points in the dataset. The second term corresponds to the rest of the dataset. For the first term, we will use the trivial bound that $\sqrt{p_i^K} \leq 1$. For the second term, 

$$ \sum_{i = v+1}^N \sqrt{p_i^K} = \sum_{i = v+1}^N \sqrt{p_i^K}\frac{\sqrt{p_{v+1}^K}}{\sqrt{p_{v+1}^K}} = \sqrt{p_{v+1}^K} B$$

\end{proof}

Using $B$ and $\Delta$, we can find a value of $K$ that bounds $|\mathbf{s}(q)|_1$ and $|\tilde{\mathbf{s}}(q)|_1$. 

\begin{theorem}
\label{thm:ChooseKtoBoundS}
Given a query $q$ and query-dependent parameters $B$ and $\Delta$, if 
$$K = \ceil[\Big]{2\frac{\log B}{\log \frac{1}{\Delta}}}$$
then
$$ p(x_v,q)^K \geq \sum_{i=v+1}^{N} p(x_i,q)^K$$
and we have the bounds 
$$ |\mathbf{s}(q)|_1 \leq v+1$$
$$ |\tilde{\mathbf{s}}(q)|_1 \leq v+1$$
\end{theorem}
\begin{proof}
Start with the inequality 
$$ \sqrt{p_v^K} \geq \sum_{i = v+1}^N \sqrt{p_i^K}$$
Now divide both sides by $\sqrt{p_{v+1}^K}$. 
$$ \left(\sqrt{\frac{p_v}{p_{v+1}}}\right)^K \geq \sum_{i = v+1}^N \frac{\sqrt{p_i^K}}{\sqrt{p_{v+1}^K}}$$
Observe that the left side is equal to $\Delta^{-K/2}$ and the right side to $B$. Thus we have
$$ \Delta^{-K/2} \geq B$$
We use the smallest integer $K$ that satisfies this inequality
$$ K \geq 2 \frac{\log B}{\log \frac{1}{\Delta}}$$
To bound the L1 norms, observe that $p_v \leq 1$ and that the summation $\leq p_v$. To get the final inequality in the theorem, start with the inequality in terms of $p_i$ rather than $\sqrt{p_i}$ and follow the same steps. The result will be $K \geq \log B / -\log\Delta $. Our choice of $K$ also satisfies this inequality (the $|\tilde{\mathbf{s}}|_1$ bound is more restrictive). 
\end{proof}

\paragraph{Equidistant Assumption:}
If we wanted to make the bound in Lemma~\ref{lem:vBBounds} or $K$ in Theorem~\ref{thm:ChooseKtoBoundS} as large as possible, we would set $B = N-v$. To have $B = N-v$, we need $p_{v+1} = p_{v+2} = ... = p_N$. Since the collision probability is a monotone function of distance, this condition means that all non-neighbors are equidistant from the query. The rationale behind our equidistant assumption is that it represents the worst possible $\Delta$-stable query. We are also motivated by~\cite{beyer1999nearest}, who also identify the equidistant case as a particularly hard instance of the near-neighbor problem. When $B = N$, the vector is minimally sparse and we rely on $K$ to do all of the work. Theorem~\ref{thm:ChooseKtoBoundS} works for any distribution of points, so we could repeat the analysis with $B < O(N)$ under stronger sparsity assumptions. However, our sketch is sublinear for stable queries even under the equidistant assumption. In the next section, we will see that the memory required by our sketch depends on $p_v$ and $\Delta$. 


\subsection{Reduce Near-Neighbor to Compressed Recovery}
\label{sec:reduceNNtoCSforCMS}


In this section, we will combine all of our results to create a near-neighbor sketch under the equidistant assumption. For simplicity, we restrict our attention to the CMS. The main challenge is to ensure that the kernel values recovered by our algorithm are within $\epsilon$ of the true ones. 


There are two sources of error: the CMS recovery and the RACE estimator. We will use $\epsilon_C$ for the CMS error and $\epsilon_E$ for the estimator error. The value of $\epsilon_C$ is determined by the CMS recovery guarantee while $\epsilon_E$ is determined by Theorem~\ref{thm:AllACEErrorBound}. We will use $M = O\left(\frac{1}{\epsilon_C} \log{\frac{N}{\delta}}\right)$ measurements for the CMS. Each measurement can differ from the true value by up to $\epsilon_E$. This situation is known as \textit{measurement noise}. For the CMS, measurement noise propagates as-is to the recovered output values. This happens because the CMS recovery procedure returns one of the cell values as its estimate for each component of $\hat{\mathbf{s}}$. If the cell values in $\widehat{\CMS}$ deviate from the true CMS values by $\leq \epsilon_E$, then the output of $\widehat{\CMS}$ deviates from the true output by $\leq \epsilon_E$. 
\begin{equation}
    s_i(q) - \epsilon_E \leq \hat{s}_i(q) \leq s_i(q) + \epsilon_E + \epsilon_C |\mathbf{s}(q)|_1 
\end{equation}
By choosing appropriate values for $\epsilon_C$ and $\epsilon_E$, we obtain a concise statement for the pointwise recovery guarantee of our estimated CMS. 



\begin{theorem}
\label{thm:PointwiseReconstructionBoundCMS}
We require 
$$ O\left(\frac{|\tilde{\mathbf{s}}(q)|_1^2 |\mathbf{s}(q)|_1}{\epsilon^3} \log\left(\frac{|\mathbf{s}(q)|_1}{\epsilon \delta}\log\left(\frac{N}{\delta}\right)\right)\log\left(\frac{N}{\delta}\right)\right)$$
ACE estimates to recover $\hat{\mathbf{s}}(q)$ such that 
\begin{equation}
s_i(q) - \frac{\epsilon}{2} \leq \hat{s_i}(q) \leq s_i(q) + \frac{\epsilon}{2}
\end{equation}
with probability $1-\delta$. 
\end{theorem}
\begin{proof}

Put $\epsilon_E = \epsilon / 4$ and $\epsilon_C = \epsilon /4 |\mathbf{s}(q)|_1$. Then the error is 
$$\mathbf{s}_i(q) - \frac{\epsilon}{4} \leq \hat{\mathbf{s}_i}(q) \leq \mathbf{s}_i(q) + \frac{\epsilon}{2} $$

To have $\epsilon_C =\epsilon /4 |\mathbf{s}(q)|_1$ with probability $1 - \delta_C$ we must have 
$$ M = O\left(\frac{|\mathbf{s}|_1}{\epsilon}\log\left(\frac{N}{\delta_C}\right)\right)$$
CMS measurements. For all measurements to have $\epsilon_E = \epsilon / 4$ with probability $1 - \delta_E$, we must have 
$$O\Big(M \frac{|\tilde{\mathbf{s}}(q)|^2_1}{\epsilon^2}\log\Big(\frac{M}{\delta_E}\Big)\Big)$$
ACE repetitions. The first requirement comes from the CMS guarantee. The second comes from Theorem~\ref{thm:AllACEErrorBound}.For both of these conditions to hold with probability $1 - \delta$, we use the union bound and put $\delta_C = \delta_E = \delta / 2$. To obtain the result, substitute $M$ into the second requirement. We can safely ignore the constant factors inside the logarithm because they are constant additive terms.
\end{proof}


\subsection{Near-Neighbor Sketch Size}
To differentiate between the $v$ and the $(v+1)^{\text{th}}$ elements of $\mathbf{s}$, we need to have $\epsilon < s_{v} - s_{v+1}$. This means that we can identify the $v$ nearest neighbors by setting $\epsilon < p_v^K - p_{v+1}^K$. 

\begin{lemma}
\label{lem:pv_asympt}
Put $K = \ceil{2\frac{\log N}{\log \frac{1}{\Delta}}}$. Then 
\begin{equation}
    \epsilon = p_v^K - p_{v+1}^K = O\left(N^{2\frac{\log{p_v}}{\log{\frac{1}{\Delta}}}}\right)
\end{equation}
\end{lemma}
\begin{proof}
$$p_v^K - p_{v+1}^K = p_v^K(1 - \Delta^K)$$
Substitute $K$:
$$ p_v^{2\frac{\log N}{\log \frac{1}{\Delta}}}(1 - \Delta^{2\frac{\log N}{\log \frac{1}{\Delta}}})$$
Recall the identity
$$ x^{\log y / \log x} = y $$
First address the $\Delta^K$ term. Observe that 
$$ \Delta^{2\frac{\log N}{\log \frac{1}{\Delta}}} = \left(\Delta^{\frac{\log N}{\log \Delta}}\right)^{-2} = N^{-2}$$
Next address the $p_v^K$ term. Observe that 
$$ p_v^{2\frac{\log N}{\log \frac{1}{\Delta}}} = \left(p_v^{\frac{\log N}{\log p_v}}\right)^{\frac{\log p_v}{\log \frac{1}{\Delta}}} = N^{2\frac{\log p_v}{\log \frac{1}{\Delta}}}$$
Put these together: 
$$p_v^K - p_{v+1}^K = N^{2\frac{\log p_v}{\log \frac{1}{\Delta}}}(1 - N^{-2}) $$
Since $(1 - N^{-2}) \to 1$ with $N$, we have that 
$$p_v^K - p_{v+1}^K = O\left(N^{2\frac{\log p_v}{\log \frac{1}{\Delta}}}\right) $$
This result may seem strange, but remember that $p_v < 1$. Therefore, $\epsilon = p_v^K - p_{v+1}^K$ is a negative power of $N$. Also, we restrict $\Delta$ to the range where $2\frac{\log{N}}{\log \frac{1}{\Delta}} > 1$. Otherwise, $K = 1$ and the lemma is unnecessary. 
\end{proof}



We are finally ready to state our main result. We assume the equidistant case and put $K = \ceil{2\frac{\log N}{\log \frac{1}{\Delta}}}$ according to Theorem~\ref{thm:ChooseKtoBoundS} and we plug the result into Theorem~\ref{thm:PointwiseReconstructionBoundCMS}. 

\begin{theorem}
\label{thm:RACEforVNN_CMS}
Given a query $q$, a dataset $\mathcal{D}$ and an LSH function that can output $r$ different values, we can construct a sketch to solve the $v$-nearest neighbor problem with probability $1-\delta$ in size
$$ O\left(v^3 N^{b_1} \log\left(\frac{v}{\delta} N^{b_2} \log{\frac{N}{\delta}}\right)\log\left({\frac{N}{\delta}}\right)\log{N}\right)$$
bits, where 
$$ b_1 = \frac{6 |\log {p_v}| + 2\log r}{\log{\frac{1}{\Delta}}}$$
$$ b_2 = \frac{2 |\log{ p_v}|}{\log{\frac{1}{\Delta}}}$$
$x_v$ is the $v^{\text{th}}$ nearest neighbor of $q$ in $\mathcal{D}$, $x_{v+1}$ is the $(v+1)^{\text{th}}$ nearest neighbor of $q$ in $\mathcal{D}$, and $\Delta = \frac{p_{v+1}}{p_v}$. 
\end{theorem}
\begin{proof}
Assume the equidistant case and put $K = \ceil{2\frac{\log N}{\log \frac{1}{\Delta}}}$. Then Theorem~\ref{thm:PointwiseReconstructionBoundCMS} states that we require 
$$ O\left(\frac{(v+1)^3}{\epsilon^3} \log\left(\frac{v+1}{\epsilon \delta}\log\left(\frac{N}{\delta}\right)\right)\log\left(\frac{N}{\delta}\right)\right)$$
ACE repetitions. The $(v+1)^3$ terms came from the bounds in Theorem~\ref{thm:ChooseKtoBoundS}. Put $\epsilon = p_v^K - p_{v+1}^K$ and use Lemma~\ref{lem:pv_asympt} to get that $\epsilon^{-1} = N^{b_2}$. The requirement is now 
$$ O\left(v^3 N^{3b_2} \log\left(N^{b_2}\frac{v}{\delta}\log\left(\frac{N}{\delta}\right)\right)\log\left(\frac{N}{\delta}\right)\right)$$
ACE repetitions. Each ACE repetition requires $r^K \log N$ bits. Apply the same trick as in Lemma~\ref{lem:pv_asympt} to get that 
$$r^K = r^{2\frac{\log N}{\log\frac{1}{\Delta}}} = N^{2\frac{\log r}{\log \frac{1}{\Delta}}}$$
The total requirement is therefore 
$$ O\left(v^3 N^{b_1} \log\left(N^{b_2}\frac{v}{\delta}\log\left(\frac{N}{\delta}\right)\right)\log\left(\frac{N}{\delta}\right)\log (N)\right)$$
bits. 
\end{proof}

Our main theorem from the main text (Corollary~\ref{cor:RACEforVNN_CMS}) is a substantially simplified version of Theorem~\ref{thm:RACEforVNN_CMS}. 
\begin{corollary}
\label{cor:RACEforVNN_CMS}
It is possible to construct a sketch that solves the exact $v$-nearest neighbor problem with probability $1-\delta$ using $ O\left(N^b \log^3\left(N\right)\right)$
bits, where 
$$ b = \frac{6 |\log {p_v}| + 2\log r}{\log{\frac{1}{\Delta}}}$$
Here, $r$ is the range of the LSH function, and $p_v$ is the collision probability of the $v^{\text{th}}$ nearest neighbor with the query.
\end{corollary}
\begin{proof}
The proof involves expanding Theorem~\ref{thm:RACEforVNN_CMS} and dropping terms. Observe that 
$$ \log\left(N^{b_1}\frac{v}{\delta}\log\left(\frac{N}{\delta}\right)\right)$$
$$ = b_1 \log N + \log v - \log \delta + \log \left(\log N - \log \delta \right)$$
This is multiplied by $v^3 N^{b_2}$, $(\log N - \log \delta)$ and $\log N$. The $N^{b_2}$ term asymptotically dominates the expression. We are left with 
$$ O\left(v^3 N^b \log^3 N \right)$$
\end{proof}

\section{Analysis Assuming Sparsity}
Here, we present results when we assume sparsity rather than near-neighbor stability. Suppose that with $K = 1$, we have $|\tilde{\mathbf{s}}(q)|_1 \leq C$ where $C$ is a query-dependent bound. In this case, we can dispense with Section~\ref{sec:sparsityconditions} as we no longer need to choose $K$ so that we get a bound for $|\mathbf{s}(q)|_1$. This greatly simplifies the analysis. In particular, we can directly apply Theorem~\ref{thm:PointwiseReconstructionBoundCMS} with the new bound for $|\tilde{\mathbf{s}}(q)|_1$.

\begin{corollary}
Given a query $q$ with $|\tilde{\mathbf{s}}(q)|_1 \leq C$, a dataset $\mathcal{D}$ and an LSH function that can output $r$ different values, we can construct a sketch to solve the $v$-nearest neighbor problem with probability $1 - \delta$ in size
$$ O\left(\frac{rC^3}{\epsilon^3} \log\left(\frac{C}{\epsilon \delta}\log\left(\frac{N}{\delta}\right)\right)\log\left(\frac{N}{\delta}\right)\log(N)\right)$$
bits, where $\epsilon = p_v - p_{v+1}$. 
\end{corollary}
\begin{proof}

Substitute the bound $|\mathbf{s}(q)|_1 \leq |\tilde{\mathbf{s}}(q)|_1 \leq C$ into Theorem~\ref{thm:PointwiseReconstructionBoundCMS}. We require 

$$ O\left(\frac{C^3}{\epsilon^3} \log\left(\frac{C}{\epsilon \delta}\log\left(\frac{N}{\delta}\right)\right)\log\left(\frac{N}{\delta}\right)\right)$$
ACE repetitions. To prove the corollary, put $\epsilon = p_v - p_{v+1}$ and multiply by $r\log N$ (the cost to store each ACE repetition). 

\end{proof}

\end{document}